\documentclass{llncs}
\usepackage{makeidx}
\usepackage{cite}
\usepackage{amsmath,amssymb,amsfonts}
\usepackage{latexsym}
\usepackage{graphicx}
\usepackage{fancybox}
\usepackage{fullpage}
\usepackage[backref]{hyperref}
\usepackage{paralist}
\usepackage{color}
\usepackage{wrapfig}
\usepackage{tikz}
\usetikzlibrary{shapes,decorations.pathreplacing}
\usepackage{setspace}
\usepackage{algorithm}
\usepackage[noend]{algpseudocode}
\usepackage[framemethod=tikz]{mdframed}
\usepackage{xspace}
\usepackage{pgfplots}
\usepackage{float}
\usepackage{framed}
\pgfplotsset{compat=1.5}

\newtheorem{observation}[theorem]{\bf Observation}

\newenvironment{proofof}[1]{\begin{trivlist} \item {\bf Proof
#1:~~}}
  {\qed\end{trivlist}}

\newcommand{\namedref}[2]{\hyperref[#2]{#1~\ref*{#2}}}
\newcommand{\thmlab}[1]{\label{thm:#1}}
\newcommand{\thmref}[1]{\namedref{Theorem}{thm:#1}}
\newcommand{\lemlab}[1]{\label{lem:#1}}
\newcommand{\lemref}[1]{\namedref{Lemma}{lem:#1}}

\newcommand{\seclab}[1]{\label{sec:#1}}
\newcommand{\secref}[1]{\namedref{Section}{sec:#1}}
\newcommand{\applab}[1]{\label{app:#1}}
\newcommand{\appref}[1]{\namedref{Appendix}{app:#1}}

\newcommand{\figlab}[1]{\label{fig:#1}}
\newcommand{\figref}[1]{\namedref{Figure}{fig:#1}}
\newcommand{\alglab}[1]{\label{alg:#1}}
\newcommand{\algoref}[1]{\namedref{Algorithm}{alg:#1}}

\newcommand{\obslab}[1]{\label{obs:#1}}
\newcommand{\obsref}[1]{\namedref{Observation}{obs:#1}}
\newcommand{\exlab}[1]{\label{ex:#1}}
\newcommand{\exref}[1]{\namedref{Example}{ex:#1}}

\newcommand{\COMMENTED}[1]{{}}

\renewcommand{\O}[1]{\ensuremath{\mathcal{O}\left(#1\right)}}

\newcommand{\eps}{\epsilon}

\newcommand{\mdef}[1]{{\ensuremath{#1}}\xspace}  
\newcommand{\myfunc}[1]{\mdef{\mathsf{#1}}}      

\DeclareMathOperator*{\polylog}{polylog}

\newcommand{\superscript}[1]{\ensuremath{^{\mbox{\tiny{\textit{#1}}}}}\xspace}
\def \th {\superscript{th}}     
\def \etal{{\it et~al.}}

\def \polylog  {\mdef{\myfunc{polylog}}}             
\renewcommand{\gcd}[1]{\mdef{\mathsf{gcd}\left(#1\right)}}
\newcommand{\HAM}[1]{\mdef{\Delta\left(#1\right)}}

\newcommand{\ignore}[1]{}

\newif\ifnotes\notestrue 
\ifnotes

\else
\newcommand{\enote}[1]{}
\newcommand{\snote}[1]{}
\fi

\hypersetup{
     colorlinks   = true,
     citecolor    = blue,
		 linkcolor		= red
}

\makeatletter
\renewcommand*{\@fnsymbol}[1]{\textcolor{red}{\ensuremath{\ifcase#1\or *\or \dagger\or \ddagger\or
 \mathsection\or \triangledown\or \mathparagraph\or \|\or **\or \dagger\dagger
   \or \ddagger\ddagger \else\@ctrerr\fi}}}
\makeatother

\providecommand{\email}[1]{\href{mailto:#1}{\nolinkurl{#1}\xspace}}

\newcommand{\figtchj}{
\begin{figure*}[!htb]
\centering
\begin{tikzpicture}[scale=0.6]

\draw (-9,0+2) -- (-1,0+2);
\draw (1,0+2) -- (9,0+2);
\foreach \x in {-9,1}{
	\draw (\x+0.1,0.8+2) -- (\x,0.8+2) -- (\x,-0.8+2) -- (\x+0.1,-0.8+2);
}
\foreach \x in {-1,9}{
	\draw (\x-0.1,0.8+2) -- (\x,0.8+2) -- (\x,-0.8+2) -- (\x-0.1,-0.8+2);
}

\draw (-9,0) -- (-1,0);
\draw (1,0) -- (9,0);
\foreach \x in {-9,1}{
	\draw (\x+0.1,0.8) -- (\x,0.8) -- (\x,-0.8) -- (\x+0.1,-0.8);
}
\foreach \x in {-1,9}{
	\draw (\x-0.1,0.8) -- (\x,0.8) -- (\x,-0.8) -- (\x-0.1,-0.8);
}

\node at (-10,2){$H_1$};
\node at (0,2){$H_2$};
\node at (-10,0){$H_3$};
\node at (0,0){$H_4$};

\node[draw,circle,inner sep=2pt,fill] at (-7,2) {};
\node[draw,circle,inner sep=2pt,fill] at (-5,2) {};
\node[draw,forbidden sign] at (-3,2) {};
\draw[decorate,decoration={brace,mirror}](-6.9,1.6) -- (-5.1,1.6);
\node at (-6,1.2){$\pi_1$};

\node[draw,circle,inner sep=2pt,fill] at (-8,0) {};
\node[draw,forbidden sign] at (-7,0) {};
\node[draw,forbidden sign] at (-6,0) {};
\node[draw,circle,inner sep=2pt,fill] at (-5,0) {};
\node[draw,circle,inner sep=2pt,fill] at (-4,0) {};
\node[draw,forbidden sign] at (-3,0) {};
\node[draw,forbidden sign] at (-2,0) {};
\draw[decorate,decoration={brace,mirror}](-7.9,-0.4) -- (-7.1,-0.4);
\node at (-7.5,-0.8){$\pi_3$};

\node[draw,circle,inner sep=2pt,fill] at (7,0) {};

\end{tikzpicture}
\caption{The dots represent candidate wildcard-periods. For any interval that has more than two dots, it follows that all dots are equally spaced after the first. The black dots represent $\mathcal{T}$ while white dots are artificially inserted to form $\mathcal{T}$, dots that follow an arithmetic sequence.}\figlab{fig:tchj}
\end{figure*}
}
\title{Periodicity in Data Streams with Wildcards
}

\author{
Funda Erg{\"{u}}n\inst{1}
\and
Elena Grigorescu\inst{2}
\and
Erfan Sadeqi Azer\inst{1}
\and
Samson Zhou\inst{2}
}
\institute{School of Informatics and Computing, Indiana University, Bloomington, IN.
\newline Email: {\tt fergun@indiana.edu, esadeqia@indiana.edu}.
\and
Department of Computer Science, Purdue University, West Lafayette, IN. 
\newline Email: {\tt elena-g@purdue.edu, samsonzhou@gmail.com}.
}

\begin{document}
\maketitle
\begin{abstract}
We investigate the problem of detecting periodic trends within a string $S$ of length $n$, arriving in the streaming model, containing at most $k$ wildcard characters, where $k=o(n)$. 
A wildcard character is a special character that can be assigned any other character. 
We say $S$ has wildcard-period $p$ if there exists an assignment to each of the wildcard characters so that in the resulting stream the length $n-p$ prefix equals the length $n-p$ suffix. 
We present a two-pass streaming algorithm that computes wildcard-periods of $S$ using $\O{k^3\,\polylog\,n}$ bits of space, while we also show that this problem cannot be solved in sublinear space in one pass.
We then give a one-pass randomized streaming algorithm that computes all wildcard-periods $p$ of $S$ with $p<\frac{n}{2}$ and no wildcard characters appearing in the last $p$ symbols of $S$, using $\O{k^3\log^9 n}$ space.
\end{abstract}

\section{Introduction}
We study the problem of detecting repetitive structure in a data stream $S$ containing a small number of \emph{wildcard characters}. 
Given an alphabet $\Sigma$ and a special \emph{wildcard} character `$\bot$'\footnote{Although wildcard characters are usually denoted with `$?$', we use $\bot$ to differentiate from compilation errors - the \LaTeX$\,$ equivalent of wildcard characters}, 
let $S\in (\Sigma\cup\,\{\bot\})^n$ be a stream that contains at most $k$ wildcards.
We can assign a value from $\Sigma$ to each wildcard character in $S$ resulting in many possible values of $S$. 
Then we informally say $S$ has \emph{wildcard-period} $p$ if there exists an assignment to each of the wildcard characters in $S$ so that the resulting string consists of the repetition of a block of $p$ characters.  

\begin{example}
\exlab{ex:wildcard:period}
The string $S=abcab\bot a\bot c\bot bc$ has wildcard-period $3$, since assigning `c' to the first wildcard character, `b' to the second wildcard character, and `a' to the third results in the string `abcabcabcabc', which consists of repetitions of the substring `abc' of length $3$.
\end{example}

The identification of repetitive structure in data has applications to bioinformatics, natural language processing, and time series data mining. 
Specifically, finding the smallest period of a string is necessary preprocessing for many algorithms, such as the classic Knuth-Morriss-Pratt \cite{KnuthMP77} algorithm in pattern matching, or the basic local alignment search tool (BLAST) \cite{Altschul90} in computational biology. 

We consider our problem in the \emph{streaming model}, where we process the input in sequential order and sublinear space. 
However in practice, some of the data may be erased or corrupted beyond repair, resulting in symbols that we cannot read, `$\bot$'. 
As a consequence, we attempt to perform pattern matching with optimistic assignments to these values. 
This motivation has resulted in a number of literature on string algorithms with wildcard characters \cite{MuthukrishnanR95, Indyk98a, ColeH02, Kalai02a, CliffordC07, HermelinR14, LewensteinNV14, GolanKP16}.

One possible approach to our problem is to generalize the exact periodicity problem, for which \cite{ErgunJS10} give a two-pass streaming algorithm for finding the smallest \emph{exact} period of a string of length $n$ that uses $\O{\log^2n}$-space and $\O{\log n}$ time per arriving symbol. 
Their results can be easily generalized to an algorithm for finding the wildcard-period of strings  using $\O{\log^2n}$-space, but at a cost of $\O{|\Sigma|^k}$ post-processing time, which is often undesirable. 
More recently, \cite{ErgunGSZ17} study the problem of $k$-periodicity, where a string is permitted to have up to $k$ permanent changes. 
The authors give a two-pass streaming algorithm that uses $\O{k^4\log^9 n}$ bits of space and runs in $\O{k^2\,\polylog\,n}$ amortized time per arriving symbol. This algorithm can be modified to recover the wildcard-period. 
We show how to do this more efficiently in \thmref{thm:twopass}.

\subsection{Our Contributions}
The challenge of determining periodicity in the presence of wildcard characters can first be approached by working toward an understanding of specific structural properties of strings with wildcard characters. 
We show in \lemref{lem:num:prints} that the number of possible assignments to the wildcard characters over all periods is ``small''. 
This allows us to compress our data into sublinear space. 
In this paper, given a string $S$ with at most $k$ wildcard characters, we show:
\begin{enumerate}
\item a two-pass randomized streaming algorithm that computes all wildcard-periods of $S$ using $\O{k^3\,\polylog\,n}$ space, {\em regardless of period length}, running in $\O{k^2\,\polylog\,n}$ amortized time per arriving symbol,
\item a one-pass randomized streaming algorithm that computes all wildcard-periods $p$ of $S$ with $p<\frac{n}{2}$ and no wildcard characters appearing in the last $p$ symbols of $S$, using $\O{k^3\,\polylog\,n}$ space, running in $\O{k^2\,\polylog\,n}$ amortized time per arriving symbol (see \appref{app:onepass}),
\item a lower bound that any one-pass streaming algorithm that computes all wildcard-periods of $S$ requires $\Omega(n)$ space even when randomization is allowed,
\item a lower bound that, for $k=o(\sqrt{n})$ with $k>2$, any one-pass randomized streaming algorithm that computes all wildcard-periods of $S$ with probability at least $1-\frac{1}{n}$ requires $\Omega(k \log n)$ space, even under the promise that the wildcard-periods are at most $n/2$.
\end{enumerate}
We remark that our algorithm can be easily modified to return the smallest, largest, or any desired wildcard-period of $S$.
Finally, we note in \appref{app:distance} several results in the related problem of determining distance to $p$-periodicity. 
We give an overview of our techniques in \secref{sec:overview}.

\subsection{Related Work}
The study of periodicity in data streams was initiated in \cite{ErgunJS10}, in which the authors give an algorithm that detlects the period of a string, using $\polylog\,n$ bits of space. 
Independently, \cite{breslauer2011real} gives a similar result with improved running time. 
Also, \cite{ElfekyAE06} studies mining periodic patterns in streams, and \cite{CrouchM11} studies periodicity via linear sketches, \cite{KoudasIM00} studies periodicity in time-series databases and online data. 
\cite{ErgunMS10} and \cite{LachishN11} study the problem of distinguishing periodic strings from aperiodic ones in the property testing model of sublinear-time computation. 
Furthermore, \cite{amir2010approximate} studies approximate periodicity in the RAM model under the Hamming and swap distance metrics. 
 
The pattern matching literature is a vast area (see \cite{ApostolicoG:1997} for a survey) with many variants.
In the data stream model, \cite{PoratP09} and \cite{CliffordFPSS16} study exact and approximate variants in offline and online settings. We use the sketches from \cite{CliffordFPSS16} though there are some other works \cite{andoni2013homomorphic,clifford2009coding,radoszewski2016streaming,porat2007improved} with different sketches for strings. 
\cite{clifford2013space} also show several lower bounds for online pattern matching problem. 

Strings with wildcard characters have been extensively studied in the offline model, usually called ``partial words''. 
Blanchet-Sadri \cite{Blanchet08} presents a number of combinatorial properties on partial words, including a large section devoted to periodicity. 
Notably, \cite{Blanchet12} gives algorithms for determining the periodicity for partial words. 
Manea \etal\,\cite{ManeaMT14} improves these results, presenting efficient time offline algorithms for determining periodicity on partial words, minimizing either total time or update time per symbol. 

Golan \etal\,\cite{GolanKP16} study the pattern matching problem with a small number of wildcards in the streaming model. 
Prior to this work, several works had studied other aspects of pattern matching under wildcards (See \cite{ColeH02},\cite{CliffordC07},\cite{HermelinR14},and \cite{LewensteinNV14}).
 
Many ideas used in these sublinear algorithms stem from related work in the classical offline model. 
The well-known KMP algorithm \cite{KnuthMP77} initially used periodic structures to search for patterns within a text. 
Galil \etal \cite{galil1983time} later improved the space performance of this pattern matching algorithm. 
Recently, \cite{gawrychowski2013optimal} also used the properties of periodic strings for pattern matching when the strings are compressed. 
These interesting properties have allowed several algorithms to satisfy some non-trivial requirements of respective models (see \cite{GolanKP16}, \cite{CliffordFPSS15} for example). 

\subsection{Preliminaries}
Given an input stream $S[1,\ldots, n]$ of length $|S|=n$ over some alphabet $\Sigma$, we denote the $i\th$ character of $S$ by $S[i]$, and the substring between locations $i$ and $j$ (inclusive) $S[i,j]$.  
We say that two strings $S,T\in\Sigma^n$ have a {\it mismatch} at index $i$ if $S[i]\neq T[i]$. 
Then the Hamming distance is the number of such mismatches, denoted $\HAM{S,T}=\Big|\{ i\mid S[i]\ne T[i]\}\Big|$.
We denote the concatenation of $S$ and $T$ by $S\circ T$. 
We denote the greatest common divisor of two integers $x$ and $y$ by $\gcd{x,y}$.

Multiple standard and equivalent definitions of periodicity are often used interchangeably.
We say $S$ has period $p$ if $S=B^{\ell}B'$ where $B$ is a block of length $p$ that appears $\ell\geq 1$ times in a row, and $B'$ is a prefix of $B$. 
For instance, $abcdabcdab$ has period 4 where $B = abcd$, and $B'=ab$. 
Equivalently, $S[x]=S[x+p]$ for all $1\le x\le n-p$. 
Similarly, the following definition is also used for periodicity.

\begin{definition}
We say string $S$ has period $p$ if the length $n-p$ prefix of $S$ is identical to its length $n-p$ suffix, $S[1,n-p]=S[p+1,n]$.
\end{definition}

More generally, we say $S$ has $k$-period $p$ (i.e., $S$ has period $p$ with $k$ mismatches) if $S[x]=S[x+p]$ for all but at most $k$ (valid) indices $x$. 
Equivalently, the following definition is also used for $k$-periodicity.

\begin{definition}
We say string $S$ has $k$-period $p$ if $\HAM{S[1,n-p],S[p+1,n]}\le k$. 
\end{definition}
The definition of $k$-periodicity lends itself to the following observation.
\begin{observation}
\obslab{obs:num:words}
If $p$ is a $k$-period of $S$, then at most $k$ substrings in the sequence of substrings $S[1,p],S[p+1,2p],S[2p+1,3p],\ldots$ can differ from the preceding substring in the sequence. 
\end{observation}
Finally, we use the following definition of wildcard-periodicity:
\begin{definition}
We say that a string $S$ has wildcard-period $p$ if there exists an assignment to the wildcard characters, so that $S[1,n-p]=S[p+1,n]$ (i.e., the resulting string has period $p$. See \exref{ex:wildcard:period}).
\end{definition}
Note that the determinism of the assignments of the characters is very important, as evidenced by \exref{ex:wildcard:wrong}.
\begin{example}
\exlab{ex:wildcard:wrong}
Consider the string $S=aaa\bot bbb$. To check whether $S$ has wildcard-period $1$, we must compare $S[1,n-1]=aaa\bot bb$ and $S[2,n]=aa\bot bbb$. 
At first glance, one might think assigning the character `$b$' to the wildcard in the prefix $S[1,n-1]$  and an `$a$' in the suffix $S[2,n]$ will make the prefix and the suffix identical. 
However, this is not a legal move; there is not a single character that the wildcard can be replaced with that makes the
above prefix and the suffix the same. 
Thus, $S$ does not have a wildcard-period of 1.
\end{example} 
The following example emphasizes the difference between $k$-periodicity and wildcard-periodicity:
\begin{example}
For $k=1$, the string $S=aaaaabbbbb$ has $k$-period $p=1$. However, to obtain wildcard-period $p=1$, at least five characters in $S$ must be changed to wildcards (for example, all of the characters `$a$' or `$b$'). 
\end{example}
Therefore, $k$-periodicity is a good notion for capturing periodicity with respect to long-term, persistent changes, while wildcard-periodicity is a good notion for capturing periodicity against a number of symbols that are errors or erasures.

We shall require data structures and subroutines that allow comparing of strings with mismatches. The below useful fingerprinting algorithm utilizes Karp-Rabin fingerprints \cite{KarpR87} to obtain general and important properties: 
\begin{theorem}\cite{KarpR87}
\thmlab{thm:kr:fingerprints}
Given two strings $S$ and $T$ of length $n$, there exists a polynomial encoding that uses $\O{\log n}$ bits of space, and outputs whether $S=T$ or $S\neq T$. Moreover, this encoding supports concatenation of strings and can be done in the streaming setting.
\end{theorem}
From here, we use the term \emph{fingerprint} to refer to this data structure.
We will also need use an algorithm for pattern matching with mismatches, which we call the $k$-mismatch algorithm.
\begin{theorem}\cite{CliffordFPSS16}
\thmlab{thm:kmismatch}
Given a string $S$ and an index $x$, there exists an algorithm which, with probability $1-\frac{1}{n^2}$, outputs all indices $i$ where $\HAM{S[1,x],S[i+1,i+x]}\le k$ using $\O{k^2\log^8 n}$ bits of space.
Moreover, the algorithm runs in $\O{k^2\,\polylog\,n}$ amortized time per arriving symbol.
\end{theorem}
Concurrent with our work, Clifford \etal\,\cite{CliffordKP17} provide a nearly-optimal solution to the $k$-mismatch algorithm, which can potentially be used in the framework of \cite{ErgunGSZ17} to immediately improve over the existing $k$-periodicity algorithms.

\section{Our Approach}
\seclab{sec:overview}
To find all the wildcard-periods of $S$, during our first pass we determine a set $\mathcal{T}$ of \emph{candidate} wildcard-periods, similar to the approach in \cite{ErgunGSZ17}, that includes all the true wildcard-periods. 
We also determine a set $\mathcal{W}$ of positions of the wildcard characters. 
By a structural result (\lemref{lem:num:prints}), we can then use the second pass to verify the candidates and identify the true wildcard-periods.

Pattern matching and periodicity seem to have a symbiotic relationship (for example, exact pattern matching and exact periodicity use each other as subroutines \cite{KnuthMP77, ErgunJS10}, as do $k$-mismatch pattern matching \cite{CliffordFPSS16} and $k$-periodicity \cite{ErgunGSZ17}). 
It feels tempting and natural to try to apply the algorithm from \cite{GolanKP16} for pattern matching with wildcards. 
Unfortunately, there does not seem to be an immediate way of doing this: the \cite{GolanKP16} algorithm searches for a wildcard-free pattern in text containing up to $k$ wildcards, while we would like to allow wildcards in the pattern \emph{and} the text. 
We instead choose to use the $k$-mismatch algorithm from \cite{CliffordFPSS16} in the first pass and obtain new structural results about possible assignments to the wildcard characters in the second pass.

In the first pass, we treat wildcards simply as an additional character. 
We let $\mathcal{T}$ be the set of indices (candidate periods) $\pi$ that satisfy 
\[\HAM{S[1,x],S[\pi+1,\pi+x]}\le 2k,\]
for some appropriate value of $x$ that we specify later. 
Note that each wildcard character can cause up to two mismatches; thus, all true wildcard-periods must satisfy the above inequality. 
We show that $\mathcal{T}$ can be easily compressed, even though it may contain a linear number of candidates. 
Specifically, we can succinctly represent $\mathcal{T}$ by adding a few additional ``false candidates'' into $\mathcal{T}$. 

If the correct assignments of the wildcards were known a priori, then the problem would reduce to determining exact periodicity. 
Unfortunately, we do not know the correct assignments to the wildcard characters prior to the data stream, so most of the difficulty lies in the guessing of assignments, bounding the total number of assignments, and storing these assignments. 
Thus, the main difference between wildcard-periodicity and both exact periodicity and $k$-periodicity is the process of verifying candidates. 
Whereas exact and $k$-periodicity can be verified by comparing the number of mismatches between the prefix and suffix of length $n-p$, wildcard-periodicity is sensitive to the correct assignments of the wildcards. 
We address this challenge by noting $\mathcal{W}$, the positions of the wildcard characters in the first pass. 
Since we also have the list of candidate wildcard-periods following the first pass, we can guess the assignments of the wildcard characters in the second pass by looking at the characters in a few select locations, as in \exref{ex:wildcard:guesses}.
\begin{example}
\exlab{ex:wildcard:guesses}
The string $S=ababa\bot ab$ has wildcard-period $p=2$. The assignment of the wildcard at position $i=6$ must be the characters at positions $i\pm p$. Note that $S[i+p]=S[8]=b$ and $S[i-p]=S[4]=b$.
\end{example}
From \exref{ex:wildcard:guesses}, we observe the following:
\begin{observation}
\obslab{obs:wildcard:guesses}
If $S$ has wildcard-period $p$ and a wildcard character is known to be at position $i$, then the assignment of the wildcard must be the character $S[i\pm ap]$, for some integer $a$, that is not a wildcard.
\end{observation} 
We show how to use \obsref{obs:wildcard:guesses} and the compressed version of $\mathcal{T}$ in the second pass to verify the candidates and output the true wildcard-periods of $S$.

We note that recent algorithmic improvements to the $k$-mismatch problem \cite{CliffordKP17} use $\O{k\log^2 n}$ space. 
Using this algorithm in place of \thmref{thm:kmismatch} as a subroutine in our algorithms improves the space usage to $\O{k^3\log^3 n}$ bits in the two-pass algorithm.

\section{Two-Pass Algorithm to Compute Wildcard-Periods}
\seclab{sec:twopass}
In this section, we provide a two-pass, $\O{k^3\log^9 n}$-space algorithm to output all wildcard-periods of some string $S$ containing at most $k$ wildcard characters. 
At a high level, we first identify a list of candidates of the periods of $S$, detected via the $k$-mismatch algorithm of \cite{CliffordFPSS16} as a black box. 
Although the number of candidates could be linear, it turns out the string has enough structure that the list of candidates can be succinctly expressed as the union of $k$ arithmetic progressions.

However, this list of candidates is insufficient in identifying the possible assignments to the wildcard characters. 
To address this issue, we explore the structure of periods with wildcards in order to limit the possible assignments for each wildcard character. 
Thus, the first pass also records $\mathcal{W}$, the positions of all wildcard characters so that during the second pass, we go over $S$ as well as the compressed data to verify the candidate periods. 

We present two algorithms in parallel to find the periods, based on their lengths. 
The first algorithm identifies all periods $p$ with $p\le\frac{n}{2}$, while the second algorithm identifies all periods $p$ with $p>\frac{n}{2}$.

\subsection{Computing Small Wildcard-Periods}
In this section, we describe a two-pass algorithm for finding wildcard-periods of length at most $n/2$. 
The first pass of the algorithm identifies a set $\mathcal{T}$ of candidate wildcard-periods in terms of indices of $S$, and maintains its succinct representation $\mathcal{T}^C$, which includes a number of additional indices. 
It also records $\mathcal{W}$, the positions of all wildcard characters. 
The second pass of the algorithm recovers each index of $\mathcal{T}$ from $\mathcal{T}^C$ and verifies whether or not the index is a wildcard-period. 
We can find the assignments of the wildcard characters in the second pass, by looking at the characters in a few locations that we determine via $\mathcal{W}$. 
We emphasize the following properties of $\mathcal{T}$ and $\mathcal{T}^C$: 
\begin{enumerate}
\item
\label{prop:1}
All wildcard-periods (possibly as well as additional candidate wildcard-periods that are false positives) are in $\mathcal{T}$.

\item
\label{prop:2}
$\mathcal{T}^C$ can be stored in sublinear space and $\mathcal{T}$ can be fully recovered from $\mathcal{T}^C$.

\item
\label{prop:3}
In the second pass, we can verify and eliminate in sublinear space candidates that are not true periods.
\end{enumerate}

In the first pass, we treat the wildcard characters as a regular, additional alphabet symbol. 
We observe that if string $S$ with such wildcards has wildcard-period $p$, there are at most $2k$ indices $i$ such that $S[i]\neq S[i+p]$, caused by the wildcard characters (the converse is not necessarily true). 
It follows that any wildcard-period $p$ must satisfy
\[\HAM{S[1,x],S[p+1,p+x]}\le 2k\]
for all $x\le n-p$, and specifically for $x=\frac{n}{2}$. 
Thus, we set $x=\frac{n}{2}$ and refer to any index $p$ that satisfies $\HAM{S[1,x],S[p+1,p+x]}\le 2k$ as a \emph{candidate wildcard-period.} 
The set of all candidate wildcard-periods forms the set $\mathcal{T}$. 
Because $\HAM{S[1,x],S[p+1,p+x]}\le 2k$ is a necessary but not sufficient condition for a wildcard-period $p$, \hyperref[prop:1]{Property 1} follows.

We give the first pass of the algorithm in full in \algoref{alg:first:pass}.

\begin{algorithm}[H]
\caption{(To determine any wildcard-period $p$ with $p\leq\frac{n}{2}$) First pass}
\alglab{alg:first:pass}
\textbf{Input:} A stream $S$ of $n$ symbols $s_i\in\Sigma\cup\,\{\bot\}$
 with at most $k$ wildcard characters $\bot$.\\
\textbf{Output:} A succinct representation of all candidate wildcard periods and the positions of the wildcard characters.
\begin{algorithmic}[1]
\State{initialize $\pi_j=-1$ for each $0\le j< 4k\log n+2$.}
\State{initialize $\mathcal{T}^C=\emptyset$.}
\For{each index $i$ (found using the $k$-mismatch algorithm) such that
\[\HAM{S\left[1,\frac{n}{2}\right],S\left[i+1,\frac{n}{2}+i\right]}\le 2k\]
}
\State{consider  $j$ for which $i$ is in the interval $H_j=\left[\frac{jn}{4(2k\log n+1)}+1,\frac{(j+1)n}{4(2k\log n+1)}\right):$}
\If{there exists no candidate $t\in\mathcal{T}^C$ in the interval $H_j$}
\State{add $i$ to $\mathcal{T}^C$.}
\Else
\State{let $t$ be the smallest candidate in $\mathcal{T}^C\cap H_j$ and either $\pi_j=-1$ or $\pi_j>0$.}
\If{$\pi_j=-1$}
\State{set $\pi_j=i-t$.}
\Else
\State{set $\pi_j=\gcd{\pi_j,i-t}$.}
\EndIf
\EndIf
\EndFor
\State{record the positions $\mathcal{W}$ of all wildcard characters.}
\end{algorithmic}
\end{algorithm}
 
Here, we show why the remaining properties for $\mathcal{T}$ and $\mathcal{T}^C$ are satisfied. 
Our algorithm divides the candidates into $\O{k\log n}$ ranges $H_1, H_2,\ldots, H_{\O{k\log n}}$ and stores the candidates in each range $H_j=\left[\frac{jn}{4(2k\log n+1)}+1,\frac{(j+1)n}{4(2k\log n+1)}\right)$ in compressed form as an
arithmetic series.

Since we use the $k$-mismatch algorithm in the first pass, we describe a structural property of the resulting list of candidates:
\begin{theorem}
\thmlab{thm:kperiod}\cite{ErgunGSZ17}
Let $p_i$ be a candidate $k$-period for a string $S$, with $p_1<p_2<\ldots<p_m$ all contained within $H_j$. 
Given the fingerprints of $S[1,n-p_1]$ and $S[p_1+1,n]$, we can determine whether or not $S$ has $k$-period $p_i$ for any $1\le i\le m$ by storing at most $\O{k^2\log n}$ additional fingerprints. 
These fingerprints represent substrings of the form $S[p_1+a\pi_j,p_1+(a+1)\pi_j-1]$, where $a>0$ is an integer and $\pi_j=\gcd{p_2-p_1,p_3-p_2,\ldots,p_m-p_{m-1}}$.
\end{theorem}
The structural property can be visualized in \figref{fig:tchj}.
\figtchj
Even though the list of candidates could be linear in size, \thmref{thm:kperiod} enforces a structure upon the list of candidates, so that an arithmetic sequence with first term $p_1$ and common difference $d$ includes all of $p_1,p_2,\ldots,p_m$. 
Thus, we can succinctly represent a superset $\mathcal{T}^C$ that contains $\mathcal{T}$ and \hyperref[prop:2]{Property 2} follows.

We now show that any wildcard period $p$ is included among the list of candidates stored by \algoref{alg:first:pass} during the first pass, and can be recovered from the list.
\begin{lemma}
\lemlab{lem:recover:indices}
If $p<\frac{n}{2}$ is a period and $p\in H_j$, then $p$ can be recovered from $\mathcal{T}^C$ and $\pi_j$.
\end{lemma}
\begin{proof}
Suppose $p\in H_j$ is a wildcard period. 
Then there exists an assignment to the wildcard characters such that $S[1,n-p]=S[p+1,n]$. 
It follows that for $i=p$, 
\[\HAM{S\left[1,\frac{n}{2}\right],S\left[i+1,\frac{n}{2}+i\right]}\le 2k,\]
so the index $i=p$ will be reported by the $k$-mismatch algorithm in the first pass. 

If at that time during Pass 1 there is no other index in $\mathcal{T}^C\cap H_j$, then $p$ will be inserted into $\mathcal{T}^C$, so $p$ can clearly be recovered from $\mathcal{T}^C$.
If there is another index $q$ in $\mathcal{T}^C\cap H_j$, then $\pi_j$ will be updated to be a divisor of $p-q$. 
Hence, $p-q$ is a multiple of $\pi_j$. 
Furthermore, any future update to $\pi_j$ will result in a value that divides the current value of $\pi_j$, due to a greatest common divisor operation. 
Thus, $p-q$ will remain a multiple of the final value of $\pi_j$, and so the set $\mathcal{T}$ at the end of the first pass will contain $p$.
\end{proof}
It remains to show that the list of candidate wildcard-periods can be verified in sublinear space in the second pass (\hyperref[prop:3]{Property 3}). 
To do this, we need a combinatorial property for periodicity on strings with wildcard characters.

\subsection{Verifying Candidates}
Recall that after the first pass, the algorithm maintains $\O{k\log n}$ succinctly represented arithmetic progressions $H_j$, corresponding to the candidate wildcard periods. 
The algorithm also maintains $\mathcal{W}$, the list of positions of wildcard characters in $S$. 
In the second pass, the algorithm must check, for each $t\in H_j$, $0\le j< 2k\log n+2$, whether $S[1,n-t]=S[t+1,n]$ for an appropriate setting of the wildcard characters. 
The challenge is computing the fingerprints of both $S[1,n-t]$ and $S[t+1,n]$ in sublinear space, especially if the number of candidates $t$ is linear. 

We first set a specific $j$ and note that for the smallest candidate $t\in H_j$, there are at most $\O{k^2\log n}$ unique substrings $S[t+1,t+\pi_j]$, $S[t+\pi_j+1,t+2\pi_j]$, $S[t+2\pi_j+1,t+3\pi_j],\ldots$. 
Since any other candidate $r\in H_j$ satisfies $r=t+a\pi_j$ for some integer $a>0$, then $S[t+1,n]$ is the concatenation 
\[S[t+1,t+\pi_j]\circ S[t+\pi_j+1,t+2\pi_j]\circ\cdots\circ S[t+(a-1)\pi_j+1,t+a\pi_j]\circ S[r+1,n].\]
Thus, by storing $\O{k^2\log n}$ fingerprints and positions, we can recover the fingerprint of the substring $S[r+1,n]$ for each $r\in H_j$. 

The second obstacle is handling wildcard characters in the computation of the fingerprints of $S[1,n-t]$ and $S[t+1,n]$. 
To address this challenge, our algorithm delays the calculation of the contribution of wildcard characters to the fingerprints until we know the assignment of the wildcard character with respect to a candidate period. 
We show that for a specific $j$, then there are at most $\O{k^2\log n}$ possible assignments for the wildcard character $S[w]=S[w\pm t]$ with respect to all candidates $t\in H_j$, across all $w\in\mathcal{W}$, where $\mathcal{W}$ is the positions of all wildcard characters recorded by \algoref{alg:first:pass}. 
Therefore, we can compute the assignment for each wildcard character with respect to a candidate period in the second pass, and then compute the fingerprint of $S[1,n-t]$ and $S[t+1,n]$.
\begin{lemma}
\lemlab{lem:num:prints}
For a given $j$, $t\in H_j$ and $w\in\mathcal{W}$, let $\sigma_t(w)$ denote the assignment of $S[w]$. 
Then $|\{\sigma_t(w)\}|=\O{k^2\log n}$.
\end{lemma}
\begin{proof}
Let $t$ be the smallest candidate in $H_j$ and $z$ be the largest candidate in $H_j$ so that $z=t+a\pi_j$ for some integer $a>0$. 
We partition $\mathcal{W}$ into $\mathcal{W}_1$, the set of indices greater than $z$, and $\mathcal{W}_2$, the set of indices no more than $z$. 
We consider the wildcard characters $w_i\in\mathcal{W}_1$, and note that the proof for $\mathcal{W}_2$ is symmetric.
Consider the $\O{k}$ sequences 
\[\begin{tabular}[!htb]{cccc}
$S[w_1-t]$ & $S[w_1-t-\pi_j]$ & $\cdots$ & $S[w_1-t-a\pi_j]$\\
$S[w_2-t]$ & $S[w_2-t-\pi_j]$ & $\cdots$ & $S[w_2-t-a\pi_j]$\\
$\vdots$ & $\vdots$ & $\ddots$ & $\vdots$ \\
$S[w_{|\mathcal{W}_1|}-t]$ & $S[w_{|\mathcal{W}_1|}-t-\pi_j]$ & $\cdots$ & $S[w_{|\mathcal{W}_1|}-t-a\pi_j]$
\end{tabular}\]
Each term in a sequence that differs from the previous term corresponds to a mismatch between $S[w_i-t-\pi_j+1,w_i-t]$, $S[w_i-t-2\pi_j+1,w_i-t-\pi_j]$, $S[w-t-3\pi_j+1,w-t-2\pi_j],\ldots$. 
For each $j$, there are at most $\O{k^2\log n}$ unique chains of substrings with length $\pi_j$ beginning at index $t+1$. 
Hence, across all $\O{k}$ sequences $S[w_i-t]$, $S[w_i-t-\pi_j]$, $S[w_i-t-2\pi_j],\ldots$, there are at most $\O{k^2\log n}$ unique characters. 
Since the assignment of $S[w_i]$ with respect to any candidate $r\in H_j$ is $S[w_i-r]=S[w_i-t-b\pi_j]$ for some integer $b>0$, then it follows that there are at most $\O{k^2\log n}$ assignments of $S[w]$ across all $w\in\mathcal{W}_1$. 
As the symmetric proof holds for $\mathcal{W}_2$, then there are at most $\O{k^2\log n}$ assignments of $S[w]$ across all $w\in\mathcal{W}$. 
\end{proof}
Thus, deciding the assignment of $S[w_i]$ with respect to a candidate $t\in H_j$ is simple:
\begin{mdframed}
For each $j$ such that $0\le j<4k\log n+2$:
\begin{enumerate}
\item
Let $t$ be the smallest candidate in $H_j$ and $z$ be the largest candidate in $H_j$ so that $z=t+a\pi_j$ for some $a>0$.
\item
For each $w\in\mathcal{W}$:
\begin{enumerate}
\item
If $w>z$, succinctly record the values of $S[w-t]$, $S[w-t-\pi_j]$, $\ldots$, $S[w-t-a\pi_j]$.
\item
If $w\le z$, succinctly record the values of $S[w+t]$, $S[w+t+\pi_j]$, $\ldots$, $S[w+t+a\pi_j]$.
\end{enumerate}
Let $r\in H_j$ so that $r=t+b\pi_j$ for some $b>0$. 
\item
The assignment of $S[w]$ with respect to $r$ is any $S[w\pm cr]$ that is not a wildcard character (where $c$ is an integer).
\end{enumerate}
\end{mdframed}
We describe the second pass in \algoref{alg:second:pass}, recalling that at the end of the first pass, the algorithm records $\O{k\log n}$ arithmetic progressions, succinctly represented, as well as the positions of all wildcard characters. 

\begin{algorithm}[H]
\caption{(To determine any wildcard-period $p$ with $p\leq\frac{n}{2}$) Second pass}
\alglab{alg:second:pass}
\textbf{Input:} A stream $S$ of symbols $s_i\in\Sigma$ with at most $k$ wildcard characters, a succinct representation of all candidate wildcard periods and the position of the wildcard characters.\\
\textbf{Output:} All wildcard-periods $p\le\frac{n}{2}$.
\begin{algorithmic}[1]
\For{each $t$ such that $t\in\mathcal{T}^C$}
\State{for each $w$ such that $w\in\mathcal{W}$, implicitly determine the value of $S[w]$ with respect to $t$.}
\State{let $j$ be the integer for which $t$ is in the interval $H_j=\left[\frac{jn}{4(2k\log n+1)}+1,\frac{(j+1)n}{4(2k\log n+1)}\right)$}
\If{$\pi_j>0$}
\Comment{$H_j$ has multiple values in $\mathcal{T}^C$}
\State{record up to $128k^2\log n+1$ unique fingerprints of length $\pi_j$, starting from $t$.}
\Else
\Comment{$H_j$ has one value in $\mathcal{T}^C$}
\State{record up to $128k^2\log n+1$ unique fingerprints of length $t$, starting from $t$.}
\EndIf
\State{check if $S[1,n-t]=S[t+1,n]$ and return $t$ if this is true.}
\EndFor
\For{each $t$ which is in interval $H_j=\left[\frac{jn}{4(2k\log n+1)}+1,\frac{(j+1)n}{4(2k\log n+1)}\right)$ for some integer $j$}
\If{there exists an index in $\mathcal{T}^C\cap H_j$ whose distance from $t$ is a multiple of $\pi_j$}
\State{check if $S[1,n-t]=S[t+1,n]$ and return $t$ if this is true.}
\EndIf
\EndFor
\end{algorithmic}
\end{algorithm}
For each arithmetic progression, there are $\O{k^2\log n}$ total possibilities for all of the wildcard characters. 
Thus, the algorithm maintains the $\O{k^3\log^2 n}$ characters corresponding to the value of all wildcard characters across all candidate positions.

We now show the ability to construct the fingerprints of $S[1,n-p]$ for any candidate period $p$.
\begin{lemma}
\lemlab{lem:recover:prints}
Let $p_i$ be a candidate $k$-period for a string $S$, with $p_1<p_2<\ldots<p_m$ all contained within $H_j$. 
Given the fingerprints of $S[1,n-p_1]$ and $S[p_1+1,n]$, we can determine whether or not $S$ has wildcard-period $p_i$ for any $1\le i\le m$ by storing at most $\O{k^2\log n}$ additional fingerprints.
\end{lemma}
\begin{proof}
Consider a decomposition of $S$ into substrings $u_j$ of length $p_i$, so that $S=u_1\circ u_2\circ u_3\circ\ldots$. 
Even though the algorithm does not record a fingerprint for each $u_j$, each index $j$ for which $u_j\neq u_{j+1}$ corresponds to at least one mismatch. 
Since the first pass searched for positions that contained at most $k$ mismatches, then it follows from \obsref{obs:num:words} that there are $\O{k}$ indices $j$ for which $u_j\neq u_{j+1}$. 
Thus, recording the fingerprints and locations of these indices $j$ suffices to build fingerprints for $S$, ignoring the wildcard characters. 
Then we can verify whether or not $p_i$ is a wildcard-period of $S$ if the assignment of the wildcard characters with respect to $p_i$ is also known.

By \thmref{thm:kperiod}, the greatest common divisor $\pi_j$ of the difference between each $p_i$ in $H_j$ is a $\O{k^2\log n}$-period. 
That is, $S$ can be decomposed $S=v\circ v_1\circ v_2\circ v_3\circ\ldots$ so that $v$ has length $p_1$, and each subsequent substring $v_i$ has length $\pi_j$. 
Then there exist at most $\O{k^2\log n}$ indices $i$ for which $v_i\neq v_{i+1}$, by \obsref{obs:num:words}. 
Ignoring wildcard characters, storing the fingerprints and positions of these indices $i$ allows the recovery of the fingerprint of $S[1,n-p_i]$ from the fingerprint of $S[1,n-p_{i-1}]$, since $p_i-p_{i-1}$ is a multiple of $\pi_j$. 
By \lemref{lem:num:prints}, we know the values of the wildcard characters with respect to $p_i$. 
Therefore, we can confirm whether or not $p_i$ is a wildcard-period.
\end{proof}
We now show correctness of the algorithm. 
\begin{lemma}
For any period $p\le\frac{n}{2}$, the algorithm outputs $p$.
\end{lemma}
\begin{proof}
Since the intervals $\{H_j\}$ cover $\left[1,\frac{n}{2}\right]$, then $p\in H_j$ for some $j$. 
It follows from \lemref{lem:recover:indices} that after the first pass, $p$ can be recovered from $\mathcal{T}$ and $\pi_j$. 
Thus, the second pass tests whether or not $p$ is a wildcard-period. 
By \lemref{lem:recover:prints}, the algorithm outputs $p$, as desired.
\end{proof}

\subsection{Computing Large Wildcard-Periods}
As in \algoref{alg:first:pass}, we would like to identify candidate periods during the first pass of the algorithm, while treating the wildcard characters as an additional symbol in the alphabet. 
Unfortunately, if a wildcard-period $p$ is greater than $\frac{n}{2}$, then it no longer satisfies 
\[\HAM{S\left[1,\frac{n}{2}\right],S\left[p+1,p+\frac{n}{2}\right]}\le 2k,\]
since $p+\frac{n}{2}>n$, and $S\left[p+\frac{n}{2}\right]$ is undefined. 
However, by treating the wildcard characters as an additional symbol, recall that $\HAM{S[1,x],S[p+1,p+x]}\le 2k$ for all $x\le n-p$. 
Then we would like to use as large an $x$ as possible while still satisfying $x\le n-p$ when choosing candidate wildcard periods $p$. 
To this effect, the observation in \cite{ErgunJS10} states that we can try exponentially decreasing values of $x$. 
Specifically, we run $\log n$ instances of the algorithm in succession, with $x=\frac{n}{2},\frac{n}{4},\ldots$. 
Note that one of these values of $x$ is the largest value as possible while still satisfying $x\le n-p$. 
As a result, the corresponding algorithm instance outputs $p$, while the other instances do not output anything.
We detail the first pass in full in \algoref{alg:big:first:pass} in \appref{app:algs}.

This partition of $[1,n]$ into the disjoint intervals $\left[1,\frac{n}{2}\right]$, $\left[\frac{n}{2}+1,\frac{n}{2}+\frac{n}{4}\right]$, $\ldots$ guarantees that any $k$-period $p$ is contained in one of these intervals. 
Moreover, the intervals $\{H^{(r)}_j\}$ partition 
\[\left[\frac{n}{2}+\frac{n}{4}+\ldots+\frac{n}{2^{r-1}},\frac{n}{2}+\ldots+\frac{n}{2^r}\right],\]
and so $p$ can be recovered from $\mathcal{T}_r^C$ and $\{\pi^{(r)}_j\}$.
We present the second pass in \algoref{alg:big:second:pass} in \appref{app:algs}.

Since correctness follows from the same arguments as the case where $p\le\frac{n}{2}$, it remains to analyze the space complexity of our algorithm.
\begin{theorem}
\thmlab{thm:twopass}
There exists a two-pass randomized algorithm using $\O{k^3\log^9 n}$ bits of space that finds the wildcard-period and runs in $\O{k^2\,\polylog\,n}$ amortized time per arriving symbol.
\end{theorem}
\begin{proof}
In the first pass, for each $\mathcal{T}_m$, we maintain a $k$-mismatch algorithm which requires $\O{k^2\log^8 n}$ bits of space, as in \thmref{thm:kmismatch}.
Since $1\le m\le\log n$, we use $\O{k^2\log^9 n}$ bits of space in total in the first pass.

In the second pass, we maintain $\O{k^2\log n}$ fingerprints for any set of indices in $\mathcal{T}_m$, and there are $\O{k\log n}$ indices in $\mathcal{T}_m$ for each $1\le m\le\log n$, for a total of $\O{k^3\log^3 n}$ bits of space. 
In addition, we store the $\O{k^2\log n}$ assignments for all the wildcard positions in each interval $H^{(r)}_j$, where $1\le r\le\log n$ and $0\le j<2k\log n+2$. 
Thus, $\O{k^3\log^9 n}$ bits of space suffice for both passes.

The running time of the algorithm is dominated by the time spent for $\log\,n$ parallel copies of $k$-mismatch algorithm in the first pass, i.e., \algoref{alg:big:first:pass}.
From \thmref{thm:kmismatch}, the $k$-mismatch algorithm runs in $\O{k^2\,\polylog\,n}$ amortized time per arriving symbol.
The rest of the algorithm consists of simple tasks like computing gcd and can be performed very quickly.
In the second pass, in total at most $\O{k^3\,\polylog\,n}$ assignments are determined and stored. Thus, the second pass runs in $\O{1}$ amortized time per arriving symbol.
\end{proof}
\section{Lower Bounds}
\seclab{sec:lb}
We first note that \cite{ErgunJS10} shows computing the period of a string in one-pass requires $\Omega(n)$ space. 
Since the problem of periodicity for strings containing wildcards is a generalization of exact periodicity, the same lower bound applies.
\begin{theorem}[Implied from Theorem 3 from \cite{ErgunJS10} and Theorem 16 from \cite{ErgunGSZ17}]
Given a string $S$ with at most $k$ wildcard characters, any one-pass streaming algorithm that computes the smallest wildcard-period requires $\Omega(n)$ space.
\end{theorem}

To show a lower bound that randomized streaming algorithm that computes all wildcard-periods of $S$ with probability at least $1-\frac{1}{n}$, even under the promise that the wildcard-periods are at most $n/2$, consider the following construction. 
Define an infinite string $1^10^11^20^21^30^3\ldots$, as in \cite{GawrychowskiMSU16}, and let $\nu$ be the prefix of length $\frac{n}{4}$. 
Define $X$ to be the set of binary strings of length $\frac{n}{4}$ with Hamming distance $\frac{k}{2}$ from $\nu$. 
For $x\in X$, let $Y_x$ be the set of binary strings of length $\frac{n}{4}$ with either $\HAM{x,y}=\frac{k}{2}$ or $\HAM{x,y}=\frac{k}{2}+1$. 
Pick $(x,y)$ uniformly at random from $(X,Y_x)$. 
Then Theorem 17 in \cite{ErgunGSZ17} shows a lower bound on the size of the sketches necessary to determine whether $\HAM{x,y}=\frac{k}{2}$ or $\HAM{x,y}=\frac{k}{2}+1$.
\begin{theorem}
\thmlab{thm:lb:mem}
\cite{ErgunGSZ17}
Any sketching function $S$ that determines whether $\HAM{x,y}=\frac{k}{2}$ or $\HAM{x,y}>\frac{k}{2}$ from $S(x)$ and $S(y)$, with probability at least $1-\frac{1}{n}$ for $k=o(\sqrt{n})$, uses $\Omega(k\log n)$ space.
\end{theorem}
Suppose Alice has $y$, along with the locations of the first $\frac{k}{2}$ positions $i$ in which $y[i]\neq x[i]$. 
Alice replaces these locations with wildcard characters $\bot$, runs the wildcard-period algorithm, and forwards the state of the algorithm to Bob, who has $x$. 
Bob then continues running the algorithm on $x\circ x\circ x$ to determine the wildcard-period of the string $S(x,y)=y\circ x\circ x\circ x$. 
Observe that:
\begin{lemma}
\lemlab{lem:lb:period}
If $\HAM{x,y}=\frac{k}{2}$, then the string $S(x,y)=y\circ x\circ x\circ x$ has period $\frac{n}{4}$. 
On the other hand, if $\HAM{x,y}=\frac{k}{2}+1$, then $S(x,y)$ has period greater than $\frac{n}{4}$.
\end{lemma}
Combining \thmref{thm:lb:mem} and \lemref{lem:lb:period}:
\begin{theorem}
For $k=o(\sqrt{n})$ with $k>2$, any one-pass randomized streaming algorithm that computes all wildcard-periods of an input string $S$ with probability at least $1-\frac{1}{n}$ requires $\Omega(k\log n)$ space, even under the promise that the wildcard-periods are at most $\frac{n}{2}$.
\end{theorem}
\section*{Acknowledgements}
We would like to thank the anonymous reviewers for their helpful comments. 
The work was supported by the National Science Foundation under NSF Awards \#1649515 and \#1619081.
\def\shortbib{0}
\bibliographystyle{alpha}
\bibliography{references}

\newcommand{\etalchar}[1]{$^{#1}$}
\begin{thebibliography}{BMRW12}

\bibitem[AEL10]{amir2010approximate}
Amihood Amir, Estrella Eisenberg, and Avivit Levy.
\newblock Approximate periodicity.
\newblock {\em Algorithms and Computation}, pages 25--36, 2010.

\bibitem[AG97]{ApostolicoG:1997}
Alberto Apostolico and Zvi Galil, editors.
\newblock {\em Pattern Matching Algorithms}.
\newblock Oxford University Press, Oxford, UK, 1997.

\bibitem[AGM{\etalchar{+}}90]{Altschul90}
Stephen~F Altschul, Warren Gish, Webb Miller, Eugene~W Myers, and David~J
  Lipman.
\newblock Basic local alignment search tool.
\newblock {\em Journal of molecular biology}, 215(3):403--410, 1990.

\bibitem[AGMP13]{andoni2013homomorphic}
Alexandr Andoni, Assaf Goldberger, Andrew McGregor, and Ely Porat.
\newblock Homomorphic fingerprints under misalignments: sketching edit and
  shift distances.
\newblock In {\em Proceedings of the forty-fifth annual ACM symposium on Theory
  of computing}, pages 931--940, 2013.

\bibitem[BG11]{breslauer2011real}
Dany Breslauer and Zvi Galil.
\newblock Real-time streaming string-matching.
\newblock In {\em Combinatorial Pattern Matching}, pages 162--172. Springer,
  2011.

\bibitem[Bla08]{Blanchet08}
Francine Blanchet{-}Sadri.
\newblock {\em Algorithmic Combinatorics on Partial Words}.
\newblock Discrete mathematics and its applications. {CRC} Press, 2008.

\bibitem[BMRW12]{Blanchet12}
Francine Blanchet{-}Sadri, Robert Mercas, Abraham Rashin, and Elara Willett.
\newblock Periodicity algorithms and a conjecture on overlaps in partial words.
\newblock {\em Theor. Comput. Sci.}, 443:35--45, 2012.

\bibitem[CC07]{CliffordC07}
Peter Clifford and Rapha{\"{e}}l Clifford.
\newblock Simple deterministic wildcard matching.
\newblock {\em Inf. Process. Lett.}, 101(2):53--54, 2007.

\bibitem[CEPR09]{clifford2009coding}
Rapha{\"e}l Clifford, Klim Efremenko, Ely Porat, and Amir Rothschild.
\newblock From coding theory to efficient pattern matching.
\newblock In {\em Proceedings of the twentieth Annual ACM-SIAM Symposium on
  Discrete Algorithms}, pages 778--784, 2009.

\bibitem[CFP{\etalchar{+}}15]{CliffordFPSS15}
Rapha{\"{e}}l Clifford, Allyx Fontaine, Ely Porat, Benjamin Sach, and
  Tatiana~A. Starikovskaya.
\newblock Dictionary matching in a stream.
\newblock In {\em Algorithms - {ESA} 23rd Annual European Symposium,
  Proceedings}, pages 361--372, 2015.

\bibitem[CFP{\etalchar{+}}16]{CliffordFPSS16}
Rapha{\"{e}}l Clifford, Allyx Fontaine, Ely Porat, Benjamin Sach, and
  Tatiana~A. Starikovskaya.
\newblock The \emph{k}-mismatch problem revisited.
\newblock In {\em Proceedings of the 27th Annual {ACM-SIAM} Symposium on
  Discrete Algorithms, {SODA}}, pages 2039--2052, 2016.

\bibitem[CH02]{ColeH02}
Richard Cole and Ramesh Hariharan.
\newblock Verifying candidate matches in sparse and wildcard matching.
\newblock In {\em Proceedings on 34th Annual {ACM} Symposium on Theory of
  Computing (STOC)}, pages 592--601, 2002.

\bibitem[CJPS13]{clifford2013space}
Rapha{\"e}l Clifford, Markus Jalsenius, Ely Porat, and Benjamin Sach.
\newblock Space lower bounds for online pattern matching.
\newblock {\em Theoretical Computer Science}, 483:68--74, 2013.

\bibitem[CKP17]{CliffordKP17}
Rapha{\"{e}}l Clifford, Tomasz Kociumaka, and Ely Porat.
\newblock The streaming k-mismatch problem.
\newblock {\em CoRR}, abs/1708.05223, 2017.

\bibitem[CM11]{CrouchM11}
Michael~S. Crouch and Andrew McGregor.
\newblock Periodicity and cyclic shifts via linear sketches.
\newblock In {\em Approximation, Randomization, and Combinatorial Optimization.
  Algorithms and Techniques - 14th International Workshop, {APPROX}, and 15th
  International Workshop, {RANDOM}. Proceedings}, pages 158--170, 2011.

\bibitem[EAE06]{ElfekyAE06}
Mohamed~G. Elfeky, Walid~G. Aref, and Ahmed~K. Elmagarmid.
\newblock {STAGGER:} periodicity mining of data streams using expanding sliding
  windows.
\newblock In {\em Proceedings of the 6th {IEEE} International Conference on
  Data Mining {(ICDM})}, pages 188--199, 2006.

\bibitem[EGSZ17]{ErgunGSZ17}
Funda Erg{\"{u}}n, Elena Grigorescu, Erfan {Sadeqi Azer}, and Samson Zhou.
\newblock Streaming periodicity with mismatches.
\newblock In {\em Approximation, Randomization, and Combinatorial Optimization.
  Algorithms and Techniques, {APPROX/RANDOM}}, pages 42:1--42:21, 2017.

\bibitem[EJS10]{ErgunJS10}
Funda Erg{\"{u}}n, Hossein Jowhari, and Mert Saglam.
\newblock Periodicity in streams.
\newblock In {\em Approximation, Randomization, and Combinatorial Optimization.
  Algorithms and Techniques, 13th International Workshop, {APPROX} 2010, and
  14th International Workshop, {RANDOM} 2010. Proceedings}, pages 545--559,
  2010.

\bibitem[EMS10]{ErgunMS10}
Funda Erg{\"{u}}n, S.~Muthukrishnan, and S{\"{u}}leyman~Cenk Sahinalp.
\newblock Periodicity testing with sublinear samples and space.
\newblock {\em {ACM} Trans. Algorithms}, 6(2):43:1--43:14, 2010.

\bibitem[Gaw13]{gawrychowski2013optimal}
Pawel Gawrychowski.
\newblock Optimal pattern matching in lzw compressed strings.
\newblock {\em ACM Transactions on Algorithms (TALG)}, 9(3):25, 2013.

\bibitem[GKP16]{GolanKP16}
Shay Golan, Tsvi Kopelowitz, and Ely Porat.
\newblock Streaming pattern matching with d wildcards.
\newblock In {\em 24th Annual European Symposium on Algorithms}, pages
  44:1--44:16, 2016.

\bibitem[GMSU16]{GawrychowskiMSU16}
Pawel Gawrychowski, Oleg Merkurev, Arseny~M. Shur, and Przemyslaw Uznanski.
\newblock Tight tradeoffs for real-time approximation of longest palindromes in
  streams.
\newblock In {\em 27th Annual Symposium on Combinatorial Pattern Matching,
  {CPM}}, pages 18:1--18:13, 2016.

\bibitem[GS83]{galil1983time}
Zvi Galil and Joel Seiferas.
\newblock Time-space-optimal string matching.
\newblock {\em Journal of Computer and System Sciences}, 26(3):280--294, 1983.

\bibitem[HR14]{HermelinR14}
Danny Hermelin and Liat Rozenberg.
\newblock Parameterized complexity analysis for the closest string with
  wildcards problem.
\newblock In {\em Combinatorial Pattern Matching - 25th Annual Symposium, {CPM}
  Proceedings}, pages 140--149, 2014.

\bibitem[IKM00]{KoudasIM00}
Piotr Indyk, Nick Koudas, and S.~Muthukrishnan.
\newblock Identifying representative trends in massive time series data sets
  using sketches.
\newblock In {\em {VLDB}, Proceedings of 26th International Conference on Very
  Large Data Bases}, pages 363--372, 2000.

\bibitem[Ind98]{Indyk98a}
Piotr Indyk.
\newblock Faster algorithms for string matching problems: Matching the
  convolution bound.
\newblock In {\em 39th Annual Symposium on Foundations of Computer Science,
  {FOCS}}, pages 166--173, 1998.

\bibitem[Kal02]{Kalai02a}
Adam Kalai.
\newblock Efficient pattern-matching with don't cares.
\newblock In {\em Proceedings of the Thirteenth Annual {ACM-SIAM} Symposium on
  Discrete Algorithms (SODA)}, pages 655--656, 2002.

\bibitem[KMP77]{KnuthMP77}
Donald~E. Knuth, James~H. {Morris Jr}., and Vaughan~R. Pratt.
\newblock Fast pattern matching in strings.
\newblock {\em {SIAM} J. Comput.}, 6(2):323--350, 1977.

\bibitem[KNW10]{KaneNW10}
Daniel~M. Kane, Jelani Nelson, and David~P. Woodruff.
\newblock An optimal algorithm for the distinct elements problem.
\newblock In {\em Proceedings of the Twenty-Ninth {ACM} {SIGMOD-SIGACT-SIGART}
  Symposium on Principles of Database Systems, {PODS}}, pages 41--52, 2010.

\bibitem[KR87]{KarpR87}
Richard~M. Karp and Michael~O. Rabin.
\newblock Efficient randomized pattern-matching algorithms.
\newblock {\em {IBM} Journal of Research and Development}, 31(2):249--260,
  1987.

\bibitem[LN11]{LachishN11}
Oded Lachish and Ilan Newman.
\newblock Testing periodicity.
\newblock {\em Algorithmica}, 60(2):401--420, 2011.

\bibitem[LNV14]{LewensteinNV14}
Moshe Lewenstein, Yakov Nekrich, and Jeffrey~Scott Vitter.
\newblock Space-efficient string indexing for wildcard pattern matching.
\newblock In {\em 31st International Symposium on Theoretical Aspects of
  Computer Science (STACS)}, pages 506--517, 2014.

\bibitem[MG82]{MisraG82}
Jayadev Misra and David Gries.
\newblock Finding repeated elements.
\newblock {\em Sci. Comput. Program.}, 2(2):143--152, 1982.

\bibitem[MMT14]{ManeaMT14}
Florin Manea, Robert Mercas, and Catalin Tiseanu.
\newblock An algorithmic toolbox for periodic partial words.
\newblock {\em Discrete Applied Mathematics}, 179:174--192, 2014.

\bibitem[MR95]{MuthukrishnanR95}
S.~Muthukrishnan and H.~Ramesh.
\newblock String matching under a general matching relation.
\newblock {\em Inf. Comput.}, 122(1):140--148, 1995.

\bibitem[PL07]{porat2007improved}
Ely Porat and Ohad Lipsky.
\newblock Improved sketching of hamming distance with error correcting.
\newblock In {\em Annual Symposium on Combinatorial Pattern Matching}, pages
  173--182, 2007.

\bibitem[PP09]{PoratP09}
Benny Porat and Ely Porat.
\newblock Exact and approximate pattern matching in the streaming model.
\newblock In {\em 50th Annual {IEEE} Symposium on Foundations of Computer
  Science, {FOCS}}, pages 315--323, 2009.

\bibitem[RS17]{radoszewski2016streaming}
Jakub Radoszewski and Tatiana~A. Starikovskaya.
\newblock Streaming k-mismatch with error correcting and applications.
\newblock In {\em 2017 Data Compression Conference, {DCC}}, pages 290--299,
  2017.

\end{thebibliography}
\appendix
\section{One-Pass Algorithm to Compute Small Wildcard-Periods}
\applab{app:onepass}
In this section, we address the problem of computing any wildcard-period $p$ that satisfies $p<\frac{n}{2}$, under the condition that no wildcard character appears in the last $p$ symbols of the string. 
As in \secref{sec:twopass}, we run two algorithms in parallel.   
The first algorithm will return any wildcard-period that satisfies $p\le\frac{n}{4}$ and the second algorithm will return any wildcard-period that satisfies $\frac{n}{4}\le p<\frac{n}{2}$. 
In the first process, we identify all indices $i$ such that $\HAM{S\left[i+1,i+\frac{n}{2}\right],S\left[1,\frac{n}{2}\right]}\le k$. 
We simultaneously track the positions of the wildcard characters and the symbol that is $i$ positions away from each wildcard character, so that we know the assignment of each wildcard character with respect to each candidate period.
Unfortunately, the second process cannot use the same paradigm, since the $k$-Mismatch algorithm reports candidate periods too late for fingerprints to be built. 
As a result, we must pre-emptively guess the candidate periods. 

\subsection{Computing Small Wildcard-Periods}
In this section, we describe the algorithm that finds any wildcard-period $p$ with $p\le\frac{n}{4}$. 
We first designate wildcard characters as unique characters and run the $k$-mismatch algorithm to find
\[\mathcal{T}=\left\{i\,\middle|i\le\frac{n}{4},\HAM{S\left[1,\frac{n}{2}\right],S\left[i+1,i+\frac{n}{2}\right]}\le k\right\}.\]
When the $k$-mismatch algorithm finds indices $i\in\mathcal{T}$, we use the fingerprints for $S\left[1,\frac{n}{2}\right]$ and $S\left[i+1,i+\frac{n}{2}\right]$ to simultaneously build the fingerprint for $S[1,n-i]$ and continue building the fingerprint for $S[i+1,n]$ respectively. 
Concurrently, we also track the positions of each wildcard character. 
For some position $w$ of a wildcard character, we identify any arbitrary non-wildcard character that is at a position $w\pmod{i}$. 
By \lemref{lem:num:prints}, we can do this in $\O{k^2\log n}$ space, and thus replace the wildcard characters in the fingerprints of $S[1,n-i]$ and $S[i+1,n]$.

The $k$-mismatch algorithm outputs $i\in\mathcal{T}$ upon reading character $i+\frac{n}{2}-1$. 
Thus for $i\le\frac{n}{4}$, it follows that $i+\frac{n}{2}-1<\frac{3n}{4}\le n-i$ so we can identify $i$ in time to build $S[1,n-i]$. 
From \thmref{thm:kperiod}, we can build each of these fingerprints from a sequence of compressed fingerprints. 

\subsection{Computing Large Wildcard-Periods}
We now describe an algorithm for identifying all wildcard-periods $p$ such that $\frac{n}{4}<p\le\frac{n}{2}$. 
Let $I_m$ be the interval $\left[\frac{n}{2}-2^m+1, \frac{n}{2}-2^{m-1}\right]$ of length $2^{m-1}$ for $1\le m\le\log n-1$ and again define a set of candidate periods:
\[\mathcal{T}_m=\left\{i\,\middle|i\in I_m,\HAM{S[1,2^m],S[i+1,i+2^m]}\le k\right\}.\]
Let $\pi_m$ be a wildcard-period of $S[1,2^m]$.
We first consider the case where $\pi_m\ge\frac{2^m}{4}$ and then the case where $\pi_m<\frac{2^m}{4}$. 

\begin{observation}
\cite{CliffordFPSS16}
\obslab{obs:apart}
If $p$ is a $k$-period for $S[1,n/2]$, then each $i$ such that 
$$\HAM{S\left[1,\frac{n}{2}\right],S\left[i+1,i+\frac{n}{2}\right]}\le\frac{k}{2}$$
must be at least $p$ symbols apart.
\end{observation} 

By \obsref{obs:apart}, if $\pi_m\ge\frac{2^m}{4}$, then $|\mathcal{T}_m|\le 4$. 
Moreover, we can detect whether $i\in\mathcal{T}_m$ by index $\frac{n}{2}-2^{m-1}+2^m$. 
On the other hand, $n-i\ge\frac{n}{2}+2^m+1$, and so we can properly build the fingerprint of $S[1,n-i]$.

Now, consider the case where $\pi_m<\frac{2^m}{4}$. 
\cite{ErgunGSZ17} show that we can compute the fingerprint of $S\left[\frac{n}{2}+1,n-i\right]$ by storing the fingerprints and positions of $\O{k^2\log n}$ substrings.

Thus, we can build the fingerprint of $S[1,n-i]$ regardless of whether $\pi_m<\frac{2^m}{4}$ or $\pi_m\ge\frac{2^m}{4}$. 
In both cases, we again simultaneously track the positions of each wildcard character. 
For some position $w$ of a wildcard character, we identify any arbitrary non-wildcard character that is at a position $w\pmod{i}$. 

By a similar reasoning to \lemref{lem:num:prints}, we can do this in $\O{k^2\log n}$ space, and thus replace the wildcard characters in the fingerprints of $S[1,n-i]$ and $S[i+1,n]$.

\begin{theorem}
There exists a one-pass algorithm that outputs all the wildcard-periods $p$ of a given string with $p\leq \frac{n}{2}$, and uses $\O{k^3\log^9 n}$ bits of space.
\end{theorem}
\begin{proof}
The $k$-mismatch subroutine that identifies candidate wildcard-periods uses $\O{k^2\log^8 n}$ bits of space. 
We also maintain $\O{k^2\log n}$ fingerprints for any set of indices in $\mathcal{T}_m$, and there are $\O{k\log n}$ indices in $\mathcal{T}_m$ for each $1\le m\le\log n$, for a total of $\O{k^3\log^3 n}$ fingerprints. 
In addition, we store the $\O{k^2\log n}$ assignments for all the wildcard positions in each interval $H^{(m)}_j$, where $1\le m\le\log n$ and $0\le j<2k\log n+2$. 
Thus, $\O{k^3\log^9 n}$ bits of space suffice.
\end{proof}

\section{Distance to $p$-Periodicity}
\applab{app:distance}
In this section, we address the problem of finding distance $\delta_p(S)$ to $p$-periodicity in a string $S$ of length $n$ containing  wildcard characters. 
That is, we find the minimum number of character changes in $S$ to obtain a string that has wildcard-period $p$. 

Suppose without loss of generality that $p$ divides $n$, so that $n=ap$ for some integer $a>0$. 
Then $S$ can be visualized as a $p\times a$ matrix $M$ so that $M_{i,j}=S[(j-1)p+i]$. 
Intuitively, $\delta_p(S)$ is the smallest number of changes to entries in matrix $M$ so that all the characters in each row are the same. 
Let $f_{-1}(M_i)$ be the frequency vector of the entries in $M_i$, the $i\th$ row of $M$, excluding both the most frequent character of $M_i$ and any wildcard characters that appear in $M_i$.
Then it follows that
\[\delta_p(S)=\sum_{i=1}^p f_{-1}(M_i).\]
It remains to estimate $f_{-1}(M_i)$ using one of several well-known techniques.
Indeed, \cite{ErgunJS10} uses several references to obtain results that directly translate to strings containing wildcard characters.
For example, \cite{ErgunJS10} use a heavy-hitter algorithm from \cite{MisraG82} to approximate $f_{-1}(M_i)$. 
We can slightly modify the technique by ignoring wildcard characters to obtain the following result:
\begin{theorem}
There exists a deterministic one-pass streaming algorithm that provides a $(1+\eps)$-approximation of $\delta_p(S)$ using $\O{\frac{p\log n}{\eps}}$ bits of space. 
\end{theorem}
Similarly, \cite{ErgunJS10} use a distinct-elements algorithm from \cite{KaneNW10} to approximate $f_{-1}(M_i)$.
Again, the technique can be modified by ignoring wildcard characters to obtain the following result:
\begin{theorem}
There exists a one-pass streaming algorithm that provides a $(2+\eps)$-approximation of $\delta_p(S)$ with probability at least $1-\delta$, using $\O{\frac{\log n}{\eps^2}\log\frac{1}{\eps}\log\frac{1}{\delta}}$ bits of space. 
\end{theorem}
\section{Full Algorithms}
\applab{app:algs}
In this section, we provide the full algorithms for finding wildcard-periods $p>\frac{n}{2}$.
We detail the first pass in full in \algoref{alg:big:first:pass}.
\begin{algorithm}[htb]
\caption{(To determine any wildcard-period $p$ if $p>\frac{n}{2}$) First pass}
\alglab{alg:big:first:pass}
\textbf{Input:} A stream $S$ of symbols $s_i\in\Sigma$ with at most $k$ wildcard characters.\\
\textbf{Output:} A succinct representation of all candidate wildcard periods and the position of the wildcard characters.
\begin{algorithmic}[1]
\State{initialize $\pi^{(m)}_j=-1$ for each $0\le j<4k\log n+2$ and $0\le m\le\log n$.}
\State{initialize $\mathcal{T}_m^C=\emptyset$.}
\For{each index $i$, let $r$ be the largest $m$ such that $\frac{n}{2}+\frac{n}{4}+\ldots+\frac{n}{2^r}\le i$.}
\State{using the $k$-mismatch algorithm, check whether 
$$\HAM{S\left[1,\frac{n}{2^r}\right],S\left[i+1,i+\frac{n}{2^r}\right]}\le 2k.$$}
\If{so, let $R=\frac{n}{2}+\frac{n}{4}+\ldots+\frac{n}{2^r-1}$.}
\State{let $j$ be the integer for which $i$ is in the interval 
\[H^{(r)}_j=\left[R+\frac{nj}{2^{r+1}(2k\log n+1)}+1,R+\frac{n(j+1)}{2^{r+1}(2k\log n+1)}\right)\]
}
\If{there exists no candidate $t\in\mathcal{T}_r^C$ in the interval $H^{(r)}_j$}
\State{add $i$ to $\mathcal{T}_{r}^C$.}
\Else
\State{let $t$ be the smallest candidate in $\mathcal{T}_r^C\cap H^{(r)}_j$ and either $\pi^{(r)}_j=-1$ or $\pi^{(r)}_j>0$.} 
\If{$\pi^{(r)}_j=-1$}
\State{set $\pi^{(r)}_j=i-t$.}
\Else
\State{set $\pi^{(r)}_j=\gcd{\pi^{(r)}_j,i-t}$.}
\EndIf
\EndIf
\EndIf
\EndFor
\State{record the positions $\mathcal{W}$ of all wildcard characters}
\end{algorithmic}
\end{algorithm}
\noindent
We present the second pass in \algoref{alg:big:second:pass}.
\begin{algorithm}[htb]
\caption{(To determine any wildcard-period $p$ with $p>\frac{n}{2}$) Second pass}
\alglab{alg:big:second:pass}
\textbf{Input:} A stream $S$ of symbols $s_i\in\Sigma$ with at most $k$ wildcard characters, a succinct representation of all candidate wildcard periods and the position of the wildcard characters.\\
\textbf{Output:} All wildcard-periods $p>\frac{n}{2}$.
\begin{algorithmic}[1]
\For{each $t$ and any $r$ such that $t\in\mathcal{T}_r^C$}
\State{Let $R=\frac{n}{2}+\frac{n}{4}+\ldots+\frac{n}{2^{r-1}}$}
\State{Let $j$ be the integer for which $t$ is in the interval 
\[H^{(r)}_j=\left[R+\frac{nj}{2^{r+1}(2k\log n+1)}+1,R+\frac{n(j+1)}{2^{r+1}(2k\log n+1)}\right)\]}
\If{$\pi^{(r)}_j>0$}
\Comment{$H^{(r)}_j$ has multiple values in $\mathcal{T}_r^C$}
\State{record up to $128k^2\log n+1$ unique fingerprints of length $\pi^{(r)}_j$, starting from $t$.}
\Else
\Comment{$H^{(r)}_j$ has one value in $\mathcal{T}_r^C$}
\State{record up to $128k^2\log n+1$ unique fingerprints of length $t$, starting from $t$.}
\EndIf
\State{check if $S[1,n-t]=S[t+1,n]$ and return $t$ if this is true.}
\EndFor
\For{each $t$ which is in interval $H^{(r)}_j=\left[R+\frac{nj}{2^{r+1}(2k\log n+1)}+1,R+\frac{n(j+1)}{2^{r+1}(2k\log n+1)}\right)$, for\\
 some integer $j$}
\If{there exists an index in $\mathcal{T}_r^C\cap H^{(r)}_j$ whose distance from $t$ is a multiple of $\pi^{(r)}_j$}
\State{check if $S[1,n-t]=S[t+1,n]$ and return $t$ if this is true.}
\EndIf
\EndFor
\end{algorithmic}
\end{algorithm}

\end{document}